\documentclass[
aip,jmp,
reprint,
a4paper,
longbibliography,
nofootinbib,
nobibnotes,onecolumn
]{revtex4-1}
\usepackage[utf8]{inputenc}
\usepackage[T1]{fontenc}
\usepackage{amssymb}
\usepackage{amsthm}
\usepackage{braket}
\usepackage{graphicx}
\usepackage{mathtools}
\usepackage{hyperref}
\usepackage{enumitem}
\hypersetup{citecolor=magenta}
\hypersetup{colorlinks=true}
\hypersetup{linkcolor=blue}
\hypersetup{urlcolor=blue}   

\newcommand\ketbra[2]{\ket{#1}\!\bra{#2}}
\newcommand\proj[1]{\ketbra{#1}{#1}}
\DeclareMathOperator\tr{tr}
\DeclareMathOperator\argmax{arg\ max}
\DeclareMathOperator\LOCC{LOCC}
\DeclarePairedDelimiter\abs\lvert\rvert
\DeclarePairedDelimiter\norm\lVert\rVert
\DeclarePairedDelimiter\floor\lfloor\rfloor

\newtheorem{theorem}{Theorem}
\newtheorem{lemma}[theorem]{Lemma}
\newtheorem*{lemma*}{Lemma}

\begin{document}

\author{Michael Gaida}
\email{michael.gaida@student.uni-siegen.de}
\affiliation{Naturwissenschaftlich--Technische Fakult\"{a}t, Universit\"{a}t Siegen, Walter-Flex-Stra{\ss}e 3, 57068 Siegen, Germany}
\author{Matthias Kleinmann}
\email{matthias.kleinmann@uni-siegen.de}
\affiliation{Naturwissenschaftlich--Technische Fakult\"{a}t, Universit\"{a}t Siegen, Walter-Flex-Stra{\ss}e 3, 57068 Siegen, Germany}

\title{Seven definitions of bipartite bound entanglement}

\begin{abstract}
An entangled state is bound entangled if one cannot combine any number of copies of the state to a maximally entangled state by using only local operations and classical communication. If one formalizes this notion of bound entanglement, one arrives immediately at four different definitions. In addition, at least three more definitions are commonly used in the literature, in particular in the very first paper on bound entanglement. Here we review critical distillation protocols and we examine how different results from quantum information theory interact in order to prove that all seven definitions are eventually equivalent. Our self-contained analysis unifies and extends previous results scattered in the literature and reveals details of the structure of bound entanglement.
\end{abstract}

\maketitle

\section{Introduction}

\begin{figure}
\centering
\includegraphics[width=0.65\linewidth, trim= 9.1cm 3.5cm 9.1cm 4.7cm, clip ]{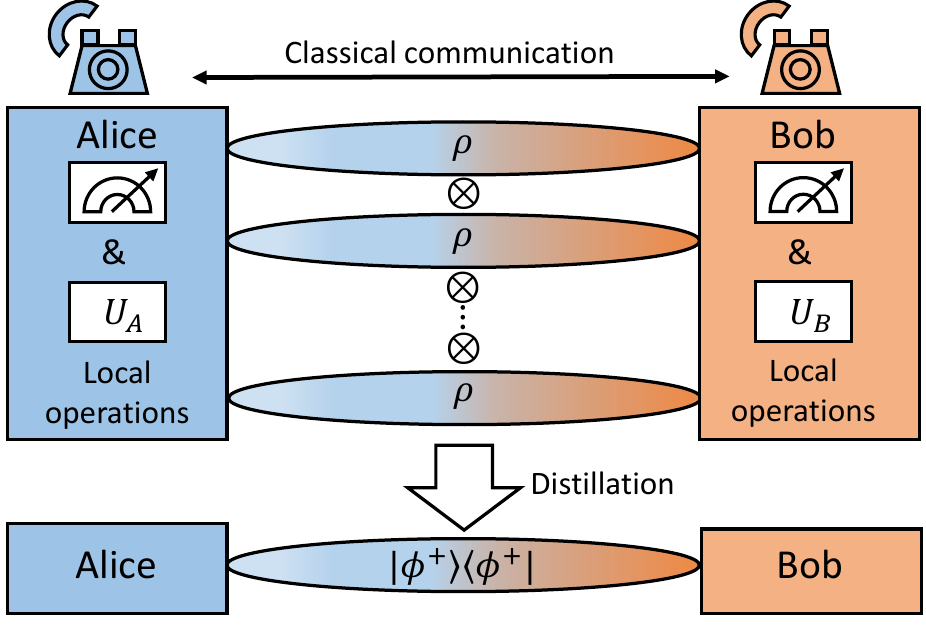}
\caption{Sketch of a general distillation procedure: Alice and Bob are in distant laboratories. They are able to perform any operation in their respective laboratories and to exchange classical information. They possess a large number $N$ of bipartite entangled states $\rho$, each with rather little entanglement. Entanglement distillation is now a procedure in which Alice and Bob consume their $N$ copies in order to produce fewer, but almost maximally entangled states (in this case, approximating $\ket{\phi^{+}}$) by means of their available operations, namely local operations and classical communication (LOCC).} 
\label{fig_sketch}
\end{figure}

Many quantum information protocols like teleportation or device-independent cryptography require quantum resources in the form of maximally highly entangled states. Unfortunately, due to interactions of the system with the environment and similar experimental imperfections, the quantum state of the system deteriorates and becomes more and more noisy, until it is insufficiently entangled for the desired application. Entanglement distillation describes the opposite procedure where more and more systems in the same entangled state can be combined into a single system in an increasingly entangled state, see Figure~\ref{fig_sketch}. The allowed operations in this process are limited to local operations and classical communication (LOCC), see also Figure~\ref{fig_sketch}. Surprisingly, not all entangled states allow such a procedure \cite{Horodecki_1998} and such states are then bound entangled.

Undistillability as an operational feature of quantum states is in some sense dual to separability. While that latter makes constraints on the resources when constructing the state, the former is based on limitations applicable when using the state. Even with this in mind, it is not true that bound entanglement is generally useless. For example, Bell inequalities can be violated using bound entangled states \cite{Vertesi2014DisprovingThePeres} and quantum key distribution is possible with bound entanglement \cite{Horodecki2005SecureKeyFrom,Horodecki2008LowDimensionalBound}. Moreover, experiments have been performed to verify bound entanglement \cite{DiGuglielmo2011ExperimentalUnconditionalPreparation,Hiesmayr2013ComplementarityRevealsBound}, even though a rigorous experimental verification is yet outstanding \cite{Sentis2018BoundEntangledStates}.
In fact, although the research on bound entanglement was initiated many decades ago, the concept is still subject of current research
\cite{lami_catalysis_2023,popp_comparing_2023,ozaydin_superactivating_2023} and connections to various fields such to many-body physics, and quantum field theory \cite{gullans_dynamical_2020, vardhan_bound_2022, klco_entanglement_2023} have been established recently. Maybe the most intriguing open question about distillability is whether the set of undistillable states is convex. Since distillability refers to what one can do with the state, it is not obvious that the set is convex and indeed, this is the case if and only if there exist bound entangled states with a negative partial transpose \cite{Shor2001NonadditivityOfBipartite}. For reviews on entanglement distillation we refer the reader to Refs.~[\onlinecite{Horodecki2009QuantumEntanglement}, \onlinecite{Clarisse2006EntanglementDistillation}]. Note that both references are roughly based on what we call ``rate-distillable'' below.

Despite of the straightforward idea underlying distillable entanglement, there are several ways to introduce a mathematical rigorous definition of the concept. Curiously, the maybe most frequently used definition, which we name ``projection distillability'' below, occurred already in the first work on bound entanglement \cite{Horodecki_1998} but this definition is not directly related to the intuitive statement above. In this work, we provide a modern view on possible definitions of bound entanglement and their classification, and we give self-contained proofs of the essential theorems in the modern language of quantum information theory.
In our definitions, we vary the criterion when distillation is achieved, but always consider LOCC transformations. This is in contrast to Ref.~[\onlinecite{Rains1999}] where Rains considers how different classes of transformations can be incorporated into the definition of bound entanglement, while focusing on a distillation criterion closely resembling ``rate distillability'' defined here.

We find two traits to classify our seven definitions, see Figure~\ref{fig_def}. The first characteristic trait is whether the definition requires a non-vanishing distillation rate, yielding an increasing number of purified states upon investing more and more instances of the state. Conversely, an asymptotically large number of states might only produce a single entangled state. The other distinctive feature is, whether the distillation protocol aims for a state close to a maximally entangled state, or, whether any entangled two-qubit state is sufficient. The arrows in Figure~\ref{fig_def} indicate our strategy to prove the equivalence of all definitions. While most proofs are straightforward, two steps require the use of sophisticated distillation protocols: We need the recurrence protocol \cite{Recurrenceoriginal} to bring any entangled two-qubit state arbitrarily close to a maximally entangled two-qubit state and the hashing protocol \cite{hashingoriginal, hashinglater} is required to distill low-entropy states with a finite rate. We traced the idea of using theses protocols for distillation back to Ref.~[\onlinecite{Horodecki1997ReductionCriterion}].

In this article, we only study bipartite bound entanglement. In fact, there is a key difference between bipartite and multipartite bound entanglement. In the multipartite case, entanglement between some of the parties can be randomly distributed to all parties in such a way, that the mixedness in the state cannot be eliminated until two or more parties meet \cite{Smolin2001FourPartyUnlockable}. This is conceptually different from the bipartite case, where there are no different ways to distribute entanglement, but still bound entanglement can be found to exist as soon as the local dimension of at least one party is four or higher \cite{entangleandfidelity, Dur2000DistillabilityAndPartial}.

This paper is structured in the following manner. Basic concepts like LOCC, local filters, and the trace norm are introduced in Section~\ref{sec_prelim}. In Section~\ref{sec_collection_defs} our seven definitions are presented. Important methods and arguments that are frequently used in the following proofs are provided in Section~\ref{sec_methods}. 
Then we provide proofs which are based on standard methods in Section~\ref{sec_straightforward}. In Sections~\ref{sec_q_to_prob} and \ref{sec_hashing} we use elaborate quantum information protocols in order to show the two remaining implications. We conclude in Section~\ref{sec_conclusions}.


\section{Preliminaries}
\label{sec_prelim}

The set of bipartite LOCC protocols is given by any quantum channel \cite{nielsen_chuang_2010} that can be implemented by the two parties Alice and Bob using only finite and local resources. More precisely, an LOCC protocol consists of a finite number of rounds and can utilize at most a finite amount of shared randomness available. In each round, Alice and Bob perform a local operation at their respective subsystem and exchange a finite amount of classical information afterwards. The total transformation of the state is then denoted by $(\Lambda) \in \LOCC$. If Alice and Bob decide whether the transformation was a success or failure at the end of the protocol, they obtain corresponding conditional transformations $\mathcal E$ and $\mathcal E'$ and we write $(\mathcal E, \mathcal E') \in \LOCC$. For a mathematically rigorous treatment of LOCC we refer the reader to Ref.~[\onlinecite{whatyoualways}]. Finally, a local filter is any operator $A$ or $B$ on Alice's or Bob's subsystem, respectively.

We use the trace distance $\norm{\tau-\rho}_1=\tr\abs{\tau-\rho} = \sum \abs{\lambda_j}$, where $\lambda_j$ are the eigenvalues of $\tau - \rho$, to quantify the distinguishability between two states $\tau$ and $\rho$. A complementary quantity is the fidelity $F(\tau, \rho)$ where we only use the case where $\tau=\proj\psi$ is pure so that $F(\tau,\rho)=\braket{\psi|\rho|\psi}$. A connection between these two notions is given by the Fuchs--van de Graaf inequalities \cite{fvdg},
\begin{equation}\label{eq_fuchs_graaf}
  1 - \braket{\psi|\rho|\psi} \leq \frac{1}{2} \norm{ \rho - \proj\psi}_1 \leq \sqrt{1 - \braket{\psi|\rho|\psi}}.
\end{equation}
We mention, that the choice of the trace distance is not completely arbitrary. Crucially, the fidelity and the trace distance can be mutually upper and lower bounded by the above relations and those bounds are independent of the dimension of the system.

A two-qubit state is maximally entangled, if it is pure and $\ket\psi=\frac 1{\sqrt 2} \ket{\alpha_1}\ket{\beta_1}+\frac 1{\sqrt 2} \ket{\alpha_2}\ket{\beta_2}$ where $\braket{\alpha_1|\alpha_2}=0$ and $\braket{\beta_1|\beta_2}=0$. We can use local unitaries to map the basis $\lbrace \ket{\alpha_1},\ket{\alpha_2}\rbrace$ and $\lbrace \ket{\beta_1},\ket{\beta_2}\rbrace$ both to $\lbrace\ket0,\ket1\rbrace$ and hence any maximally entangled state is equivalent to $\ket{\Phi^+}=\frac1{\sqrt 2}(\ket{00}+\ket{11})$ under local unitaries.


\section{Definitions of bipartite distillability}
\label{sec_collection_defs}

\begin{figure}
\centering
\includegraphics[width=0.65\linewidth, trim= 10.5cm 5.5cm 7.3cm 2.2cm, clip]{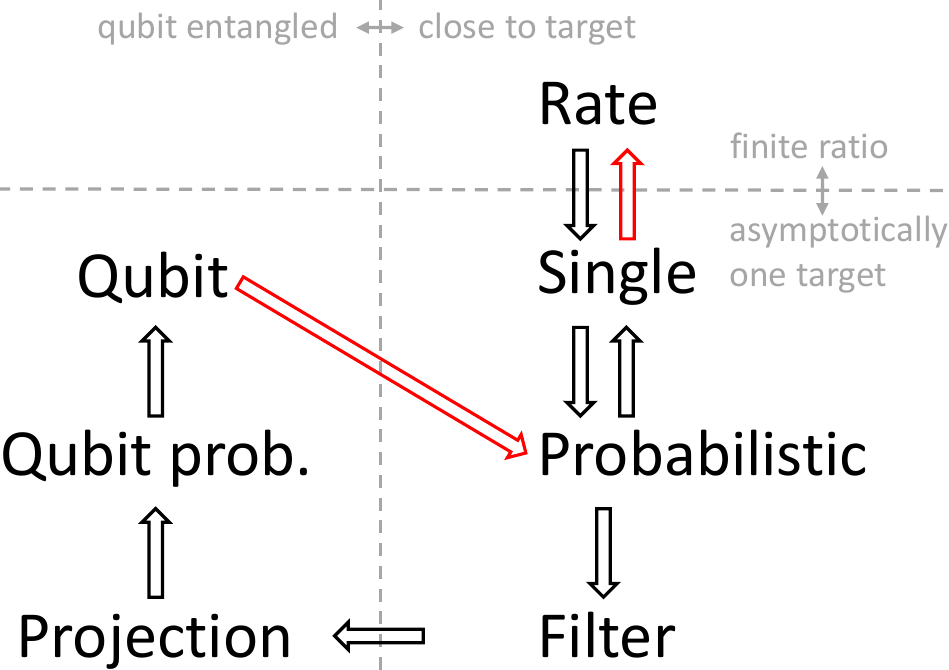}
\caption{Relations between the seven definitions of distillability. The arrows correspond to the direction of the proof that we use in order to show the equivalence of all definitions. Black arrows correspond to straightforward proofs and red arrows denote proofs based on elaborate quantum information protocols. The horizontal dashed line separates rate distillability from definitions allowing an asymptotically small yield. The vertical dashed line separates definitions where the distilled state is close to a maximally entangled target state and definitions where the distilled state is any entangled two-qubit state. }
\label{fig_def}
\end{figure}

Here we present rigorous definitions of the aforementioned concept of distillability. A bipartite state is distillable if the two parties can approximate via LOCC a maximally entangled state by consuming a finite number of copies of the original state. By convention one chooses a maximally entangled two-qubit state and we use here, without loss of generality, see Section~\ref{sec_prelim}, $\sigma_+=\proj{\Phi^+}$ with $\ket{\Phi^+}=\frac1{\sqrt 2}(\ket{00}+\ket{11})$ as target state. 

A bipartite state $\rho$ is
\begin{enumerate}[align=left,font=\em]
\item[rate-distillable] if there is a positive rate $R$ such that
\begin{equation}
  \inf \set{ \norm{ \sigma_+^{\otimes M} - \Lambda(\rho^{\otimes N}) }_1 |
  N,M \in \mathbb{N},\; (\Lambda) \in \LOCC,\; M \geq RN } = 0.
\end{equation}

Rate distillability reflects the strongest notion of distillation. Investing $N$ copies of $\rho$ we obtain at least $RN$ purified copies with an asymptotically vanishing error, that is, the number of purified copies grows with a fixed rate, as the number of invested copies grows.

\item[single-distillable] if
\begin{equation}
  \inf \set{ \norm{ \sigma_+ - \Lambda( \rho^{\otimes N}) }_1 |
  N \in \mathbb{N},\; (\Lambda) \in \LOCC } = 0.
\end{equation}

For single distillability we drop the requirement of a finite rate and  require just a single purified copy in the asymptotic limit.

\item[probabilistic-distillable] if
\begin{equation}
  \inf \set{ \norm{ \sigma_+ - \tfrac1p\mathcal{E}( \rho^{ \otimes N}) }_1 |
  N \in \mathbb{N},\; (\mathcal{E},\mathcal E') \in \LOCC,\;
  p=\tr[ \mathcal E(\rho^{\otimes N})] > 0 } = 0.
\end{equation}
Probabilistic distillability is a loosened version of single distillability that allows the LOCC protocol to not always succeed.

\item[filter-distillable] if
\begin{equation}
  \inf \set{ \norm{ \sigma_+ - (A \otimes B)\rho^{\otimes N}(A \otimes B)^\dagger }_1 |
  N \in \mathbb{N},\; A,B \text{ local filters} } = 0.
\end{equation}

Filter distillability allows us to map a state arbitrarily close to a maximally entangled state by means of local filters.
\end{enumerate}

The remaining definitions do not require to get close to a specific maximally entangled state. 

A bipartite state $\rho$ is
\begin{enumerate}[align=left,font=\em]

\item[projection-distillable] if there is an integer $N$ and local projectors $\Pi_A$, $\Pi_B$ onto two-dimensional subspaces such that $( \Pi_A \otimes \Pi_B ) \rho^{ \otimes N} ( \Pi_A \otimes \Pi_B )$ is an entangled state, up to normalization.

For projection distillability it is hence sufficient to project $N$ copies of $\rho$ to an entangled two-qubit state.

\item[qubit probabilistic-distillable] if there is an integer $N$ and an LOCC protocol $(\mathcal{E},\mathcal E')$ such that $\mathcal{E} ( \rho^{\otimes N} )$ is an entangled two-qubit state, up to normalization.

Qubit-probabilistic distillability is analogous to projection distillability, but uses LOCC protocols instead of projections.

\item[qubit-distillable] if there is an integer $N$ and an LOCC protocol $\Lambda$ such that $\Lambda ( \rho^{\otimes N} )$ is an entangled two-qubit state.

For qubit distillability, we require a deterministic LOCC protocol, rather than the probabilistic protocol from the previous definition.
\end{enumerate}

These definitions capture an intuitive notion of distillable entanglement and are motivated from formulations of distillability in literature.
In Ref.~[\onlinecite{Horodecki_1998}] the calculations are based on (probabilistic) qubit distillability and expressions resembling filter and projection distillability can also be found in this paper. Rate distillability is particularly intuitive and appears, for example, in Refs.~[\onlinecite{Rains1999}, \onlinecite{Leuchs2019}].
To our knowledge the definitions of single and probabilistic distillability have not been mentioned before. We introduce them in order to break up the proofs into elemental steps.

A logical classification of the seven definitions is presented in Figure~\ref{fig_def}. Noticeably, the upper left quadrant in this classification is vacant. The reason is that there is no natural candidate from literature or an emergent candidate from the proofs that would fit there. Nonetheless, it is possible to contrive according definitions, see Appendix~\ref{sec_qbit_rate} for an example. Further definitions slightly outside of this scheme are discussed in Appendix~\ref{sec_more_definitions}, and \ref{App_convrate}, but these additional definitions are still equivalent to the seven definitions presented in this section. In Appendix~\ref{App_perfect} we also give an example of a naive, but ill-behaved notion of distillable entanglement.


\section{Recurrent methods}
\label{sec_methods}

In this section we introduce and review important methods that are frequently needed building blocks in the main theorems. We show how probabilistic LOCC protocols can be approximated by deterministic ones, mention an important bound on the fidelity of two-qubit states and discuss the twirling method.

\subsection{Iterative postselection procedure}
\label{post_selection}

\begin{lemma}\normalfont\label{lemm_postselection}
Consider states $\rho$ and $\rho'$ and an LOCC protocol $(\mathcal{E},\mathcal E')$ such that $\rho' = \frac{1}{p} \mathcal{E} (\rho)$ with success probability $p = \tr[ \mathcal{E}( \rho)] > 0$. Then there exists a sequence of LOCC protocols $(\Lambda_n)_n$ such that
\begin{equation}
  \lim_{n\to\infty} \norm{ \Lambda_n ( \rho^{\otimes n} ) - \rho' } =0.
\end{equation}
\end{lemma}
\begin{proof}
We consider $n$ copies of $\rho$ and apply the LOCC protocol $(\mathcal{E},\mathcal E')$ to the first copy. If the procedure is successful, we obtain $\rho' =\frac1p\mathcal E(\rho)$, discard the remaining copies of $\rho$, and terminate the protocol. If the procedure is not successful, this copy of $\rho$ is transformed to the unwanted state $\tau=\frac1{1-p}\mathcal E'(\rho)$ and we apply $\mathcal{E}$ to the next copy of $\rho$. After a maximum of $n$ unsuccessful applications of $\mathcal{E}$, we discard all copies of $\tau$ except one. This yields the sequence of $\LOCC$ protocols $( \Lambda_n)_n$ with
\begin{equation}
  \Lambda_n ( \rho^{\otimes n} ) = [ 1 - (1-p)^n] \rho' + (1-p)^n \tau.
\end{equation}
It follows that
\begin{equation}
  \lim_{n\to\infty}\norm{ \Lambda_n ( \rho^{\otimes n} ) - \rho' } =
  \lim_{n\to\infty}[(1-p)^n\norm{ \tau - \rho'}] = 0
\end{equation}
due to $1-p<1$.
\end{proof}


\subsection{Maximally entangled two-qubit states}
\label{app_maxent}

The maximal fidelity of a pure two-qubit product state $\ket\psi=\ket{\alpha}\ket{\beta}$ with any maximally entangled state can be directly verified to be
\begin{equation}
  \max_{\ket\alpha,\ket\beta}\abs{\braket{\alpha,\beta|\Phi^+}}^2 = \frac 12.
\end{equation}
Since the separable states are the convex hull of the pure product states \cite{werner_sep} and the fidelity $\braket{\psi|\rho|\psi}$ is linear in $\rho$, it follows that the maximal fidelity of a separable two-qubit state with a maximally entangled state is $\frac12$. Conversely, the fidelity of an entangled two-qubit state $\rho$ with some maximally entangled state is larger than $\frac 12$. This fact was first proved in Ref.~[\onlinecite{Horodecki1997ReductionCriterion}].


\subsection{Twirling}
\label{sec_Twirling}

Any bipartite two-qubit state $\rho$ with $F=\braket{\Phi^+|\rho|\Phi^+}$ can be mapped via LOCC to the Werner-type state
\begin{equation}\label{eq_Werner}
  W_F = F \proj{\Phi^{+}} + \frac{1-F}3( \proj{\Psi^{+}} + \proj{\Phi^{-}} + \proj{\Psi^{-}} ),
\end{equation}
where $\ket{\Phi^+}$, $\ket{\Phi^-}$, $\ket{\Psi^+}$, and $\ket{\Psi^-}$ denote the Bell-states as usual.
This can be achieved by Alice and Bob performing the same random unitary \cite{werner_sep} and a subsequent correction,
\begin{align}
  W_F = (\sigma_x\otimes\sigma_z)\left(\int (U\otimes U) \rho (U\otimes U)^\dag \mathrm{d}U\right) (\sigma_x\otimes\sigma_z)^\dag,
\end{align}
where integration is with respect to the Haar measure of $\mathrm{SU}(2)$ and $\sigma_x$ ($\sigma_z$) is the $x$ ($z$) Pauli matrix. However, using this method, Alice and Bob need to choose a random unitary, which is from an infinite set and hence in principle requires an unlimited amount of shared randomness or communication.

This latter subtlety can be avoided by using only a finite set of unitaries \cite{hashinglater}. These unitaries can be understood as a three-step process. In the first step Alice and Bob perform randomly the same Pauli matrix $\openone$, $\sigma_x$, $\sigma_y$, or $\sigma_z$, yielding already a state in Bell-diagonal form. In the next step, a random choice of the matrices $\eta_0=\openone$, $\eta_1=SH$, and $\eta_2=\sigma_3HS$ is applied, where $H$ is the Hadamard gate and $S$ the phase gate. Finally Alice and Bob apply the same correction as above. In summary,
\begin{align}
  W_F=\frac1{12}\sum_{\nu,\mu}(\sigma_x\otimes\sigma_z)(\eta_\nu\otimes \eta_\nu)
  (\sigma_\mu\otimes \sigma_\mu)\rho (\sigma_\mu\otimes \sigma_\mu)^\dag
  (\eta_\nu\otimes \eta_\nu)^\dag(\sigma_x\otimes\sigma_z)^\dag.
\end{align}


\section{Implications with straightforward proofs}
\label{sec_straightforward}

Given the definitions, it is plausible that rate distillability is the strongest and can easily be loosened to probabilistic distillability with a stopover at single distillability. We first give a pedestrian's proof that rate-distillable states are also single-distillable.
\begin{theorem}\normalfont
\label{theo_rate_single}
Rate distillability implies single distillability.
\end{theorem}
\begin{proof}
Let $\rho$ be rate-distillable and $\epsilon > 0$. Consequently, there are integers $N\ge M \ge 1$ and an LOCC protocol $\Lambda$ such that $\norm{\sigma_+^{\otimes M} - \Lambda(\rho^{\otimes N})}_1 < \epsilon $. By tracing out $M -1$ subsystems $X$ and using that the partial trace cannot increase the trace distance \cite{nielsen_chuang_2010}, we obtain
\begin{equation}
  \norm{ \tr_X ( \sigma_+^{\otimes M } ) - \tr_X [ \Lambda (\rho^{\otimes N} ) ] }_1 \leq \norm{ \sigma_+^{\otimes M} - \Lambda(\rho^{\otimes N}) }_1 \leq \epsilon.
\end{equation}
Clearly, the map $\tr_X$ can be trivially performed by LOCC. This means that we found an integer $N$ and an LOCC protocol $\tilde{\Lambda} =\tr_X\circ\Lambda$ such that $\norm{ \sigma_+ - \tilde{\Lambda}( \rho^{\otimes N})}_1 \leq \epsilon$ and the assertion follows.
\end{proof}

Clearly, if a state is single-distillable then it is also probabilistic-distillable because the former is a special case of the latter. The converse is also true.
\begin{theorem}\normalfont
Single distillability is equivalent to probabilistic distillability.
\end{theorem}
\begin{proof}
It remains to show the that probabilistic distillability implies single distillability. Let $\rho$ be probabilistic-distillable. Therefore, for every $\epsilon > 0$ we find an integer $N$ and an LOCC protocol $(\mathcal{E},\mathcal E')$ such that
\begin{equation}
  \norm{ \frac{1}{p} \mathcal{E} ( \rho^{\otimes N} ) - \sigma_+ }_1 \leq \frac\epsilon2
\end{equation}
with the success probability $p = \tr [ \mathcal{E}( \rho^{\otimes N} ) ]>0$. Using the strategy for postselection explained in Section~\ref{post_selection}, according to Lemma~\ref{lemm_postselection} we find an integer $n$ and a deterministic LOCC protocol $\Lambda$ such that
\begin{equation}
  \norm{ \Lambda( \rho^{\otimes nN} ) - \frac{1}{p} \mathcal{E} ( \rho^{\otimes N} ) }_1 \leq \frac\epsilon2.
\end{equation}
Using the triangle inequality completes the proof as $\norm{\Lambda( \rho^{\otimes nN} ) - \sigma_+}_1 \leq \epsilon$.
\end{proof}

Filter-distillable is different from the previous definitions in that it does not directly involve LOCC operations but rather uses the local filter $A\otimes B$. Still, any probabilistic-distillable state is already filter distillable.
\begin{theorem}\normalfont
Probabilistic distillability implies filter distillability.
\end{theorem}
\begin{proof}
Suppose $\rho$ is probabilistic-distillable. Then for $\epsilon > 0$ there is an integer $N$ and an LOCC protocol $(\mathcal{E},\mathcal E')$ such that
\begin{equation}
  \norm{ \sigma_+ - \frac{1}{p} \mathcal{E} ( \rho^{\otimes N} )}_1 \leq \frac{\epsilon^2}2,
\end{equation}
where $p = \tr[ \mathcal{E} ( \rho^{\otimes N} )] >0 $. Applying the first Fuchs--van de Graaf inequality in Eq.~\eqref{eq_fuchs_graaf} yields
\begin{equation}\label{eq_FDW_c}
  1 - \frac{1}{p} \tr [ \sigma_+ \mathcal{E} (\rho^{\otimes N})] \leq \frac{\epsilon^2}{4}.
\end{equation}
Now we consider the Kraus decomposition \cite{nielsen_chuang_2010} of the operation $\mathcal{E}$,
\begin{equation}
  \mathcal{E}( \rho^{\otimes N} ) = \sum_k K_k \, \rho^{\otimes N} K_k^\dagger = \sum_k p_k \tau_k.
\end{equation}
Note that the operators $K_k$ can be chosen to be of product form, due to $(\mathcal{E}, \mathcal{E}') \in \LOCC$. Using the linearity of the trace in Eq.~\eqref{eq_FDW_c} we obtain
\begin{equation}
  \sum_k q_k \tr ( \sigma_+ \tau_k) \geq 1- \frac{\epsilon^2}{4},
\end{equation}
where $q_k={p_k}/{p}$. From $\sum q_k=1$ we see that for $k_0=\argmax_k \tr(\sigma_+ \tau_k)$ we have $\tr ( \sigma_+ \tau_{k_0}) \geq 1- \frac{\epsilon^2}{4}$. The second Fuchs--van de Graaf inequality in Eq.~\eqref{eq_fuchs_graaf} gives us now
\begin{equation}
  \norm{\sigma_+ - \tau_{k_0}}_1 \leq 2 \sqrt{1-\left(1-\frac{\epsilon^2}{4}\right)} = \epsilon.
\end{equation}
Finally, we can define $A$, $B$ such that
\begin{equation}
  A \otimes B = \frac{1}{q_{k_0}} K_{k_0}
\end{equation}
which is possible because $K_{k_0}$ is of product form. Therefore, for every $\epsilon >0$ we find local filters $A$ and $B$ and some $N$ such that
\begin{equation}
  \norm{ \sigma_+ - (A \otimes B) \rho^{\otimes N} (A \otimes B)^\dagger }_1 \le \epsilon.
\end{equation}
\end{proof}

We now make the transition from the definitions where we approximate the entangled state $\sigma_+$ to definitions where we transform to any unspecified entangled two-qubit state.
\begin{theorem}\normalfont
Filter distillability implies projection distillability.
\end{theorem}
\begin{proof}
The proof consists mainly of an argument given in Ref.~[\onlinecite{Horodecki_1998}]. For completeness we discuss the steps in detail. First note that the set of separable states is closed \cite{Horodecki2009QuantumEntanglement} and hence, the set of entangled states is open in the set of all states. Let now $\rho$ be filter-distillable. Then there is a trace norm $\epsilon$-neighborhood around $\sigma_+$ that contains only entangled states for sufficiently small $\epsilon >0$. Now we use the filter distillability of $\rho$ to find an integer $N$ and operators $A$, $B$ such that
\begin{equation}
  \norm{ ( A \otimes B ) \rho^{\otimes N} ( A \otimes B )^\dagger - \sigma_+ }_1 < \epsilon.
\end{equation}
Consequently, $(A \otimes B) \rho^{\otimes N} (A \otimes B)^\dagger$ is entangled because it lies in the entangled $\epsilon$-neighborhood of $\sigma_+$. Since $A$ and $B$ map onto qubit spaces, we can express them as
\begin{equation}
  A = \ketbra{0}{\psi_A} + \ketbra{1}{\phi_A} \text{ and }
  B = \ketbra{0}{\psi_B} + \ketbra{1}{\phi_B},
\end{equation}
with not necessarily normalized vectors $\ket{\psi_{A/B}}$ and $\ket{\phi_{A/B}}$. We now define the projector $\Pi_A$ which projects onto the subspace spanned by $ \ket{\psi_{A}}$, $\ket{\phi_{A}}$. Then
\begin{equation}
  (A \otimes B) \rho^{\otimes N} (A \otimes B)^\dagger=
  ( A \otimes B ) ( \Pi_A \otimes \Pi_B ) \rho^{\otimes N} ( \Pi_A \otimes \Pi_B ) ( A \otimes B )^\dagger.
\end{equation}
The separable operator $(A \otimes B)$ cannot entangle a separable state. Thus, the projected state $( \Pi_A \otimes \Pi_B ) \rho^{\otimes N} ( \Pi_A \otimes \Pi_B )$ is already entangled, up to normalization.
\end{proof}

One can always extend the projectors from projection distillability to a probabilistic LOCC protocol, yielding qubit-probabilistic distillability.
\begin{theorem}\normalfont
Projection distillability implies qubit-probabilistic distillability.
\end{theorem}
\begin{proof}
Suppose $\rho$ is projection-distillable and let $N$, $\Pi_A$, and $\Pi_B$ be from the definition. We complete the projections $\Pi_A$ and $\Pi_B$ to local measurements. With a probability of $p = \tr( \Pi_A \otimes \Pi_B \, \rho^{\otimes N }) > 0$ these measurements yield the projected state $(\Pi_A \otimes \Pi_B) \rho^{\otimes N} ( \Pi_A \otimes \Pi_B)$, which is, up to normalization, entangled by assumption.
\end{proof}

Next, one can upgrade any probabilistic LOCC protocol to an imperfect deterministic LOCC protocol, by applying the protocol from Section~\ref{post_selection}.
\begin{theorem}\normalfont
Qubit probabilistic distillability implies qubit distillability.
\end{theorem}
\begin{proof}
Given a qubit probabilistic-distillable state $\rho$ we have an integer $N$ and an $\LOCC$ protocol $( \mathcal{E}, \mathcal{E}')$ such that $\tau = \frac{1}{p} \mathcal{E} ( \rho^{\otimes N })$ is entangled with $p = \tr[ \mathcal{E} ( \rho^{\otimes N})] > 0$. Using the postselection strategy presented in Section~\ref{post_selection}, we find according to Lemma~\ref{lemm_postselection} an integer $n$ and a deterministic $\LOCC$ protocol $\Lambda$ for every $\epsilon > 0$ such that $\norm{ \Lambda ( \rho^{\otimes nN} ) - \tau }_1 \leq \epsilon$ holds. Since the entangled states are topologically open in the set of states \cite{Horodecki2009QuantumEntanglement}, we can map into an $\epsilon$-neighborhood of $\tau$ where every state is entangled.
\end{proof}


\section{Qubit implies probabilistic}
\label{sec_q_to_prob}

In this section we cross the vertical boundary in Figure~\ref{fig_def} from left to right, that is, we transition back from a notion of distillability which produces any qubit entangled state to a notion where a maximally entangled state is approximated. This transition requires the recurrence protocol.

The recurrence protocol is described in Refs.~[\onlinecite{Horodecki2009QuantumEntanglement},~\onlinecite{Recurrenceoriginal}]. Since the actual implementation is rather technical, we focus on the relevant outcome and refer the reader to the aforementioned sources. A single step in the protocol consumes two copies of the Werner-type state $W_F$, see Eq.~\eqref{eq_Werner}, and produces the state $W_{F'}$ with $F'>F$. Remember that $F$ is the fidelity of $W_F$ with respect to $\ket{\Phi^+}$.
\begin{lemma}\normalfont\label{lem_recc}
For every $\frac12 < F < 1$ there exists an LOCC protocol $(\mathcal E,\mathcal E')$ with $\frac1p\mathcal E(W_F^{\otimes 2})=W_{F'}$ and $p=\tr[\mathcal E(W_F^{\otimes 2}) ]>\frac5{18}$ such that
\begin{equation}
  1 > F' > g(F) = \frac{10 F^2 - 2F +1}{8 F^2 - 4F +5} > F.
\end{equation}
\end{lemma}
\begin{proof}
Alice and Bob perform on their respective qubits an operation with the single Kraus operator $K=\ketbra{0}{00}+\ketbra{1}{11}$, that is, $\mathcal E(\rho)= (K\otimes K)\rho(K\otimes K)^\dag$. One verifies that the resulting state is again Bell-diagonal with
\begin{align}
  F'&=\braket{\Phi^+|\frac1p\mathcal{E} (W_F^{\otimes 2} )|\Phi^+}=\frac1p \frac1{18}(10F^2-2F+1)> F \quad\text{and}\\
  p&=\tr [ \mathcal{E} ( W_F^{\otimes 2} ) ]=\frac1{18}(8F^2-4F+5)>\frac{5}{18}.
\end{align}
Applying twirling, see Section~\ref{sec_Twirling}, yields then the state $W_{F'}$.
\end{proof}

We mention that in the original recurrence protocol \cite{Recurrenceoriginal}, $p$ is larger by a factor of two, because in the first step one can apply an operation with two successful outcomes.

Iterating the recurrence protocol we can achieve now any target fidelity $F_\text{target}<1$.

\begin{lemma}\normalfont\label{lemm_recurrence_iterative}
For every two-qubit state $\rho$ with $F=\braket{\Phi^+|\rho|\Phi^+}>\frac12$ and any target fidelity $F_\text{target}<1$ there is an integer $N$ and an LOCC protocol $(\mathcal{E},\mathcal E')$ such that $\frac{1}{p} \mathcal{E} ( \rho^{\otimes 2N} ) = W_{F'}$, where $F' \ge F_\text{target} $ and $p = \tr[ \mathcal{E} ( \rho^{\otimes 2N}) ] > 0$.
\end{lemma}
\begin{proof}
First we apply twirling to the state to achieve the Bell-diagonal state $W_F$. Then, iterating the recurrence protocol, we can achieve any fidelity $F_\text{target}<1$ in a finite number of steps: Consider the sequence of iterated fidelities $F_{n+1} = g(F_n)$ with $g$ given according to Lemma~\ref{lem_recc}. This sequence converges to unity, since it is strictly monotonously increasing with $g(1)=1$. Therefore, $1>F_N\ge F_\mathrm{target}$ can be reached for a finite number of steps $N$. The probability to succeed in the $N$ iterations is lower bounded by $(5/18)^N$ and hence the total probability of success is also finite.
\end{proof}

We are now in the position to prove the main result the of this section:

\begin{theorem}\normalfont
Qubit distillability implies probabilistic distillability.
\end{theorem}
\begin{proof}
Let $\rho$ be qubit-distillable and $\epsilon > 0$. Consequently, we can map $N$ instances of $\rho$ to an entangled two-qubit state $\tau$ by means of $\LOCC$. In Ref.~[\onlinecite{Horodecki1997ReductionCriterion}] it is shown that the fidelity of an entangled two-qubit state with some maximally entangled state $\sigma'$ is always larger than $\frac12$, see also Section~\ref{app_maxent}. Since we can always change the local bases such that $\sigma'=\sigma_+$ \cite{nielsen_chuang_2010}, we can assume, without loss of generality, that $\sigma'=\sigma_+$. After twirling $\tau$ into Bell-diagonal form, see Section~\ref{sec_Twirling}, we use the recurrence protocol \cite{Recurrenceoriginal}. According to Lemma~\ref{lemm_recurrence_iterative}, for any $\epsilon>0$ we find an integer $M$ and an LOCC protocol $(\mathcal{E}, \mathcal E')$ such that
\begin{equation}
  \tr(\sigma_+ \tfrac{1}{p} \mathcal{E} ( \tau^{\otimes M} ) ] \geq 1 - \frac{\epsilon^2}{4}.
\end{equation}
with success probability $p=\tr[\mathcal E(\tau^{\otimes M})]>0$. The second Fuchs--van de Graaf inequality in Eq.~\eqref{eq_fuchs_graaf} yields $\norm{\frac{1}{p} \mathcal{E} ( \tau^{\otimes M}) - \sigma_+}_1 \leq \epsilon$ which completes the proof.
\end{proof}

\section{Single implies rate}
\label{sec_hashing}

It remains to cross the horizontal line in Figure~\ref{fig_def} in order to transition from single distillability to rate distillability. The crucial step hereby is the hashing protocol, but in order to apply the protocol we first need a Bell-diagonal state with low von Neumann entropy $S(\rho)=-\tr[\rho \log_2(\rho)]$.
\begin{lemma}\normalfont\label{lem_single_hashable}
If $\rho$ is single-distillable, then for any $0<h<S(\rho)$ one can produce a Bell-diagonal state $\rho'$ with $S(\rho')<h$ using LOCC and a finite number of copies of $\rho$.
\end{lemma}
\begin{proof}
Given a single-distillable state $\rho$, for every $\epsilon >0$ there is an integer $N$ and an LOCC protocol $\Lambda$ such that $\tau= \Lambda(\rho^{\otimes N})$ has fidelity $\tr(\sigma_+\tau)\ge1-\epsilon$, where we used Eq.~\eqref{eq_fuchs_graaf}. After twirling, see Section~\ref{sec_Twirling}, the resulting state $\rho'$ is Bell-diagonal while having maintained the same fidelity with $\sigma_+$. Since $\sigma_+$ has a von Neumann entropy of $S(\sigma_+) = 0$ and $S$ is continuous, we have that $S(\rho') <h$ for $\epsilon$ sufficiently small.
\end{proof}
Now, we can proceed to apply the hashing protocol \cite{hashingoriginal, hashinglater}. This counter-intuitive procedure uses ideas from Shannon's coding theorem in order to pick $m$ maximally entangled pure states from the $n$ mixed input states with high probability. For this, only about $\frac n2(1+S(\rho))$ states have to be sacrificed.

The hashing protocol purifies $n$ copies of a Bell-diagonal two-qubit state $\rho$ with von Neumann entropy of $S(\rho) < 1$. The protocol can be divided into two parts: An iterative random quantum measurement is performed on some of the qubit pairs, gathering parity information. A classical protocol determines the Bell-states of the remaining qubits from the gained parity information.

For clarification, we outline the rough idea of the classical part first: We encode the index of the four Bell states in two classical bits. Therefore the mixture of Bell states $\rho^{\otimes n}$ can be written as a mixture of random messages of $n$ bits pairs of an information source $X$ with Shannon entropy $H(X) = S(\rho)$, where $H(X)= -\sum_x \mathbb P(X=x)\log_2\mathbb P(X=x)$. From Shannon's coding theorem it is known that we can compress a message from this source reliably to a message of $n H(X)$ bits, provided $n$ is chosen sufficiently large. Vice versa, $n H(X)$ bits of information about the message are sufficient to guess the complete message, where the error probability vanishes asymptotically. This allows us to sacrifice $n H(X)/\eta$ Bell states ($2 n H(X)/\eta$ bits from the message) to learn the state of the remaining $m=n (1- H(X)/\eta)$ qubit pairs. Based on the acquired knowledge, the remaining qubit pairs are known Bell states and can be transformed by means of LOCC to $\sigma_+^{\otimes m}$. Here $\eta$ is the number of bits which we can acquire by measuring a qubit pair. Although a measurement on a qubit pair can give two bits of information, the limitations of LOCC allow us only to obtain one bit. To incorporate the effect of finite $n$, we further half the number of remaining qubit pairs. According to Refs.~[\onlinecite{hashingoriginal}, \onlinecite{hashinglater}], we have the following statement:

\begin{lemma}\normalfont\label{lemm_hashing_exact}
Assume a source of Bell diagonal states $\rho$ with $S(\rho)<1$ and choose a target fidelity $F<1$. Then there exists an LOCC protocol that consumes $n$ copies of $\rho$ while producing a state $\rho'$ of $m$ qubit pairs $\rho'$, such that $\tr(\sigma_+^{\otimes m}\rho')\ge F$ while $m\ge \frac n2[1-S(\rho)]$.
\end{lemma}

The steps for proving Lemma~\ref{lemm_hashing_exact} are explained in Appendix~\ref{app_hash}.

\begin{theorem}\normalfont
Single distillability implies rate distillability.
\end{theorem}
\begin{proof}
Given a single-distillable state $\rho$, we can use the LOCC protocol from Lemma~\ref{lem_single_hashable} to map $N$ copies of $\rho$ onto one copy of a Bell diagonal state $\rho'$ with $h=S( \rho') < 1$. Therefore $\rho'$ is suitable for applying the hashing protocol. According to Lemma~\ref{lemm_hashing_exact}, for any $0<\epsilon<1$ there exists an LOCC protocol $\Lambda$ mapping $n$ qubit pairs to $m$ qubit pairs such that $\tr[\sigma_+^{\otimes m}\Lambda(\rho^{\prime \otimes n})]\ge 1-\frac{\epsilon^2}4$ while $m\ge \frac n2(1-h)$. This yields in terms of the trace distance,
\begin{equation}
  \norm{\sigma_+^{\otimes m}-\Lambda(\rho^{\prime \otimes n})]}\le 2\sqrt{1-\left(1-\frac{\epsilon^2}4\right)}=\epsilon.
\end{equation}
Combining the rate $\frac mn\ge \frac 12(1-h)$ with the the number of copies $N$ used to obtain $\rho'$, we can distill entanglement by means of LOCC with a fixed total rate of at least $(1-h)/(2N)>0$. Note that here $N$ is a fixed constant throughout the whole procedure while $n$ and $m$ are determined by the error tolerance $\epsilon$.
\end{proof}


\section{Discussion}
\label{sec_conclusions}

The seven definitions in Sec.~\ref{sec_collection_defs} have been selected by three criteria: occurrence in literature, having an intuitive definition, and capturing essential properties of bound entanglement. For example, the transition from single distillability to rate distillability summarizes what can be achieved with the hashing protocol. Similarly, the transition between qubit distillability and probabilistic distillability summarizes the recurrence method. From the aspect of usability, filter distillability and projection distillability are the simplest definitions, because they do not require one to consider the convoluted class of LOCC operations. It is even surprising that all of the definitions are equivalent. In particular, rate distillability requires an asymptotic constant rate of distillation while qubit probabilistic distillability only aims for distillation of a single, arbitrarily entangled qubit state with nonvanishing probability.

The big repertory of equivalent definitions allows one to choose the best suited definition for any problem.
For example, projection distillability immediately implies that the set of undistillable states is closed: If $\Pi_A\otimes\Pi_B$ maps $N$ copies of $\rho$ to an entangled two qubit state $\sigma$, then this is also the case for a small neighborhood around $\rho$, since the entangled qubit states are an open set. Similarly, a positive partial transpose implies undistillability, which is a direct consequence of the distillability of all entangled qubit states \cite{entangleandfidelity}. However, one of the central questions regarding bound entanglement still remains unanswered, namely whether there exist bound entangled states which do not have a positive partial transpose, or equivalently \cite{Shor2001NonadditivityOfBipartite}, whether the set of undistillable state is convex.

Summarizing, we presented seven different definitions of distillable entanglement, covering many notions used in the literature as well as introducing new nuances to complete the systematic picture outlined in Figure~\ref{fig_def}. We structured the definitions such that they are aligned with the structure of the proofs, making a unified picture on the matter available. We also provided a treatment of the recurrence and hashing protocols, so that the current article can serve as a self-contained entry point for new research on the topic.


\begin{acknowledgements}
We thank
Dagmar Bruß,
Adán Cabello,
Hermann Kampermann,
Lisa Weinbrenner, and
Zeng-Peng Xu
for discussions.
We are particularly indebted to Paweł Horodecki for advice and pointing out key literature.
This work was supported by the Deutsche Forschungsgemeinschaft (DFG, German Research Foundation, project numbers 447948357 and 440958198), the Sino-German Center for Research Promotion (Project M-0294), the ERC (Consolidator Grant 683107/TempoQ) and the German Ministry of Education and Research (Project QuKuK, BMBF Grant No.~16KIS1618K).
MG acknowledges support from the House of Young Talents of the University of Siegen.
\end{acknowledgements}

\appendix


\section{Further definitions of distillability}

\subsection{Qubit rate distillability}
\label{sec_qbit_rate}
The upper left quadrant in Figure~\ref{fig_def} of the main text is vacant and one may wonder whether there is a suitable definition of distillability that fits there. This would correspond to a notion of qubit distillability with a finite rate. To our knowledge, no such definition has been used in literature, but for completeness we contrive here a possible definition and show that it is equivalent to qubit distillability.

A state $\rho$ is \emph{qubit rate-distillable} if there is a positive rate $R$ and a threshold $N_0$, such that for all $N\ge N_0$ one can find an LOCC protocol $\Lambda$ with the property that $\Lambda[\rho^{\otimes N}]=\Sigma$ is a state on $M\ge RN$ qubit pairs and the reduced state $\Sigma_k$ for each qubit pair $k$ is entangled.

\begin{theorem}
    \normalfont 
    Qubit rate distillability is equivalent to qubit distillability.
\end{theorem}
\begin{proof}
Let $\rho$ be qubit rate-distillable. For $N_0$ copies of $\rho$ and the corresponding LOCC protocol $\Lambda$ yielding $\Lambda[\rho^{N_0}]=\Sigma$, any of the reduced states $\Sigma_k$ is qubit entangled and hence $\rho$ is qubit-distillable.

Conversely, let $\rho$ be qubit-distillable. Then there is an integer $N'$ and an LOCC protocol $\Lambda$ such that $\Lambda[ \rho^{\otimes N'}]=\tau$ is two-qubit entangled. For $N\ge N'$, we apply $\Lambda$ to $M$ chunks of $N^\prime$ copies of $\rho$, with $M$ maximal, that is, $M=\floor{ N/N' }$. Therefore, we obtain from $\rho^{\otimes N}$ the state $\Sigma=\tau^{\otimes M}$ for which trivially each $\Sigma_k$ is qubit entangled. Finally, choosing the rate $R=1/(2N')$, it follows that $M\ge RN$ for all $N \ge N_0= 2N'$.
\end{proof}

\subsection{Extended rate distillability}
\label{sec_more_definitions}

The definition of rate distillability requires the distillation of entangled qubit pairs with a fixed rate. Instead, it is also possible to require the creation of a high-dimensional maximally entangled state, in particular of
\begin{equation}
  \ket{\psi_{d}^{+}} = \frac1{\sqrt{d}}\sum_{k=0}^{d-1} \ket{kk}.
\end{equation}
Note that $\ket{\psi_{2^M}^+} = \ket{\Phi^+}^{\otimes M}$.
Embracing this idea, a bipartite state $\rho$ is \emph{extended rate-distillable} if there is a positive rate $R$, such that
\begin{equation}
  \inf\lbrace\norm{\proj{\psi_d^+}-\Lambda(\rho^{\otimes N})}_1 \mid d,N\in\mathbb N,\;
  (\Lambda)\in \LOCC,\; \log_2 d\ge RN \rbrace = 0.
\end{equation}
Any rate-distillable state is also extended rate-distillable, since the former reduces to the latter if we restrict in the infimum $d$ to be a power of 2. The converse is also true.
\begin{theorem}\normalfont
Extended rate distillability is equivalent to rate distillability.
\end{theorem}
\begin{proof}
It remains to show that extended rate-distillability implies rate-distillability. For this we note that with the definition
\begin{equation}
  \delta(\rho|d,R) =
  \inf\lbrace\norm{\proj{\psi_d^+}-\Lambda(\rho^{\otimes N})}_1 \mid N\in\mathbb N,\;
  (\Lambda)\in \LOCC,\; RN\le \log_2 d \rbrace,
\end{equation}
we can rewrite extended rate-distillability to $\inf\lbrace \delta(\rho|d,R) \mid d\in\mathbb N\rbrace=0$. Since we can always use an LOCC protocol $\Lambda$ on two copies of $\rho^{\otimes N}$, we clearly have $\delta(\rho|d,R)\ge \delta(\rho|d^2,R)$ and thus we can equivalently require
\begin{equation}
  \liminf_{d\to\infty}\delta(\rho|d,R)=0.
\end{equation}
In comparison, rate-distillability is expressed by the same equation, but with $d$ limited to be a power of two. In particular, it is sufficient to show that this restriction of $d$ does not change the value of the limit. While this implication is intuitively plausible, we here give an explicit proof.

We first fix a qubit ratio $0<\omega<1$ and consider only $M_d=\floor{ \omega\log_2 d }$ qubit pairs. We partition the $d$-dimensional local space into $\kappa=\floor{ d2^{-M_d} }$ local subspaces of dimension $2^{M_d}$ and a remaining subspace. This gives rise to the local projections
\begin{equation}
  \Pi_j = \sum_{\ell=0}^{2^{M_d}-1} \ketbra{\ell}{j 2^{M_d}+\ell}
\end{equation}
of $j$'th local subspace onto the $M_d$ local qubits, where $j=0,1,\dotsc,\kappa-1$. A distillation protocol $\Lambda$ is now followed by the LOCC protocol $\Gamma$ which implements the projections $T_j=\Pi_j\otimes\Pi_j$. By writing $\Lambda(\rho^{\otimes N})=\proj{\psi^+_d}+\Delta$, we have
\begin{equation}\begin{split}
  \Gamma[\Lambda(\rho^{\otimes N})]
  &=\sum_j T_j \proj{\psi^+_d}T^\dag_j +G(\proj{\psi^+_d}) + \Gamma(\Delta)\\
  &=\frac{\kappa 2^{M_d}}{d} \sigma_+^{\otimes {M_d}} + G(\proj{\psi^+_d}) + \Gamma(\Delta)
\end{split}\end{equation}
Here, $\sigma_+=\proj{\Phi^+}$ and $G$ is the map that occurs if the measurement outcomes of Alice and Bob do not coincide or if either of the two parties obtains an outcome from the remaining subspace. Since $\Delta$ is traceless, so is $\Gamma(\Delta)$ and hence $\tr[G(\proj{\psi^+_d})]=1-\frac{\kappa 2^{M_d}}d$. Using the triangular inequality and the fact that channels cannot increase the trace distance \cite{nielsen_chuang_2010}, we find
\begin{equation}\begin{split}
  \inf_\Lambda \norm{ \sigma_+^{\otimes M_d} - \Gamma[\Lambda( \rho^{\otimes N})]}_1
  &= \norm{\left(1-\frac{\kappa 2^{M_d}}d\right) \sigma_+^{\otimes M_d}-G(\proj{\psi^+_d})- \Gamma(\Delta)}_1\\
  &\le 2\left(1-\frac{\kappa 2^{M_d}}d\right) + \norm{\proj{\psi^+_d}-\Lambda(\rho^{\otimes N})}_1.
\end{split}\end{equation}
Due to $d-\kappa 2^{M_d}< d^\omega$, we immediately obtain
\begin{equation}
  \delta(\rho|2^{M_d},R)< 2d^{\omega-1}+\delta(\rho|d,R).
\end{equation}
Hence, if we have a sequence $d_1,d_2,\dotsc$ with $d_i\to\infty$ and for which $\lim_i \delta(\rho|d_i,R)= \liminf_d \delta(\rho|d,R)$, then also $2^{M_{d_i}}\to\infty$ and $\lim_i \delta(\rho|2^{M_{d_i}},R)=\liminf_d \delta(\rho|d,R)$ due to $\omega<1$. This proofs the above intuitive implication.
\end{proof}


\subsection{Convergent rate distillability}
\label{App_convrate}
The definitions in Section~\ref{sec_collection_defs} of the main text represent various approaches to bound entanglement. Despite of their variety, they are either based on an infimum over the one-norm or require the production of some entangled state. A different strategy is to require a convergent rate as the number of copies increases. The following definition is based on definitions used in Refs.~[\onlinecite{Rains1999}, \onlinecite{hashinglater}].

For a given state $\rho$, we consider a sequence of LOCC protocols $(\Lambda_n)_n$ such that there exists a corresponding sequence of dimensions $(d_n)_n$ obeying
\begin{equation}
  \lim_{n\to\infty}\norm{\proj{\psi_{d_n}^{+}}-\Lambda_n( \rho^{\otimes n} )}_1 = 0.
\end{equation}
A bipartite state $\rho$ is \emph{convergent rate-distillable}, if
\begin{equation}\label{eq_distillable_ent}
  \limsup_{n\to\infty} \frac{\log_2 (d_n)}{n} > 0
\end{equation}
can be achieved for an appropriate sequence $(\Lambda_n)_n$. This definition is again equivalent to any other notion of distillable entanglement.
\begin{theorem}\normalfont
Convergent rate distillability is equivalent to extended rate distillability.
\end{theorem}
\begin{proof}
Let $\rho$ be extended rate-distillable. Then for some $R>0$ there is a sequence of LOCC protocols $(\Lambda_n)_n$ and sequences $(N_n)_n$, $(d_n)_n$ with $\log_2 d_n\ge RN_n\ge Rn$, such that
\begin{equation}
  \lim_{n\to\infty}\norm{\proj{\psi_{d_n}^{+}} - \Lambda_n(\rho^{\otimes N_n}) }_1= 0.
\end{equation}
Indeed, the condition $N_n\ge n$ can always be satisfied, since using more copies cannot deteriorate the distillation protocol. Consequently, $(\Lambda_n)_n$ is an appropriate sequence of LOCC protocols for convergent rate distillability with
\begin{equation}
  \limsup_{n\to\infty} \frac{\log_2 (d_n)}{n} \ge \lim_{n\to\infty} \frac{\log_2(d_n)}{N_n} \geq R > 0.
\end{equation}
This completes the first part of the proof.

For the converse, let $\rho$ be convergent rate-distillable. Therefore, there is a sequence $ (\Lambda_n)_n$ and dimensions $(d_n)_n$ obeying Eq.~\eqref{eq_distillable_ent}. We can now find a sequence $(n_k)_k$ with
\begin{equation}
  \lim_{k \to \infty} \frac{\log_2(d_{n_k})}{n_k} >0
\end{equation}
which tells us that there is a rate $R>0$ with $\log_2(d_{n_k}) \geq R n_k$ for $k\ge k_0$ and $k_0$ sufficiently large. Then
\begin{equation}
  \inf\lbrace\norm{\proj{\psi^+_{d_{n_k}}}-\Lambda_{n_k}(\rho^{\otimes n_k})}_1 \mid k\ge k_0 \rbrace=0,
\end{equation}
and hence the state is also extended rate distillable.
\end{proof}


\subsection{Impossibility of perfect distillation}
\label{App_perfect}

One might also ask whether it is sensible to require that a pure entangled state can be distilled after consuming only a finite number of copies of $\rho$. This is not the case. Assuming the contrary, there must local filters $A$ and $B$ such that
\begin{equation}\label{eq_exact_filter}
  ( A \otimes B ) \rho^{\otimes N } ( A \otimes B )^\dag = \sigma'
\end{equation}
for some pure entangled state $\sigma'$. We assume now that our state is $\rho_\epsilon = (1 - \epsilon) \sigma' + \epsilon \frac{\openone}{\tr(\openone)}$ with $0<\epsilon < 1$.
Then
\begin{equation}
  \rho_\epsilon^{\otimes N } = ( 1 - \epsilon^N ) \tau + \epsilon^N \left(\frac\openone{\tr(\openone)}\right)^{\otimes N }
\end{equation}
where we only need to know that $\tau$ is some state.
Applying the local filters we obtain from Eq.~\eqref{eq_exact_filter},
\begin{equation}
  ( 1 - \epsilon^N ) ( A \otimes B ) \tau ( A \otimes B )^\dag + \epsilon^N \tr(\openone)^{-N } ( AA^\dagger \otimes B B^\dagger )=\sigma'.
\end{equation}
The left hand side is a convex combination of two states which must combine to the pure state $\sigma'$. Hence each of the states must be equal to $\sigma'$. But this is a contradiction, since $\tr(\openone)^{-N } AA^\dag\otimes BB^\dag$ is a product state while $\sigma'$ is entangled.


\section{Proof of Lemma~\ref{lemm_hashing_exact}}\label{app_hash}
In this appendix we provide the steps to prove Lemma~\ref{lemm_hashing_exact}, see Refs.~[\onlinecite{hashingoriginal}, \onlinecite{hashinglater}] for further details.

\begin{lemma*}
Assume a source of Bell diagonal states $\rho$ with $S(\rho)<1$ and choose a target fidelity $F<1$. Then there exists an LOCC protocol that consumes $n$ copies of $\rho$ while producing a state $\rho'$ of $m$ qubit pairs $\rho'$, such that $\tr(\sigma_+^{\otimes m}\rho')\ge F$ while $m\ge \frac n2[1-S(\rho)]$.
\end{lemma*}

For the proof, we start by writing the Bell-diagonal state as
\begin{equation}
  \rho = \mathbb E[\,\proj{\psi(X)}\,] = \sum_{x=0}^3 \mathbb P(X=x)\proj{\psi(x)},
\end{equation}
where $\ket{\psi(0)}=\ket{\Phi^+}$, $\ket{\psi(1)}=\ket{\Psi^+}$, $\ket{\psi(2)}=\ket{\Phi^-}$, and $\ket{\psi(3))}=\ket{\Psi^-}$ and $X$ is an appropriate random variable taking values $0,1,2,3$. Correspondingly, we have
\begin{equation}
  \rho^{\otimes n} = \sum_{x_1,\dotsc,x_n} \mathbb P(\vec X=\vec x)\proj{\psi(\vec x)}
\end{equation}
with $\vec X=(X_1,\dotsc,X_n)$ identical and independently distributed (i.i.d.) random variables and 
$\ket{\psi(\vec x)}=\ket{\psi(x_1)}\ket{\psi(x_2)}\dotsm\ket{\psi(x_n)}$. If we would know that the value of $\vec X$ is $\vec x$, then the state under this knowledge would reduce to the pure entangled state $\proj{\psi(\vec x)}$. Instead of determining the value of $\vec X$ with certainty, we rather aim to guess the value $\vec X$ correctly with high probability. For this it is sufficient to determine the value only if it is in the $\epsilon$-typical set
\begin{equation}
  A_{n,\epsilon} = \set{ \vec x | -\epsilon< -\tfrac1n\log_2[\mathbb P(\vec X=\vec x)] - H(X) \le \epsilon }.
\end{equation}
As we discuss in Appendix~\ref{sec_eps_typical}, the probability $\mathbb P[\vec X\notin A_{n,\epsilon}]$ converges to zero as $n\to\infty$, that is, the value of $\vec X$ is almost certainly within the $\epsilon$-typical set if $n$ is sufficiently large. At the same time, the number of elements in the $\epsilon$-typical set is upper bounded by $2^{nH(X)+n\epsilon}$. Hence it suffices to ``learn'' $nH(X)+n\epsilon$ bits of $\vec X$ to determine the value of $\vec X$.

This ``learning'' is performed by means of an LOCC protocol. The quantum-mechanical implementation of this protocol is described in Ref.~[\onlinecite{hashinglater}]. To understand its effect, we assume that $\vec X=\vec x$ and we write $\mathbf x$ for the bit vector of length $2n$ corresponding $\vec x$ by encoding each entry $x_i=0,1,2,3$ as bit pair $00,01,10,11$. For a given bit string $\mathbf s$, the LOCC protocol yields the parity bit $t(\mathbf s;\vec x)=\sum_j \mathbf s_j\mathbf x_j\mod 2$. In addition, the vector $\vec x$ gets transformed in a certain way, $\vec x\mapsto f(\mathbf s;\vec x)$, and also looses one of its entries. After iterating this procedure with $r\le n$ bit strings $\mathbf s^{(1)},\dotsc,\mathbf s^{(r)}$, one obtains the parity bit string $\mathbf t(\vec x)$ and the final transformed vector $\vec x^{(r)}$. Here, $\mathbf t_k(\vec x)= t(\mathbf s^{(k)};\vec x^{(k-1)})$ and $\vec x^{(k)}=f(\mathbf s^{(k)}; x^{(k-1)})$ with $\vec x^{(0)}=\vec x$. If each bit string $\mathbf s^{(k)}$ is chosen randomly from the uniform bit string distribution, then the probability that two vectors $\vec x$ and $\vec y$ yield the same parity vector but different final vectors is given by
\begin{equation}\label{eq_parity_prob}
  \mathbb P[ {\mathbf t}(\vec x) = {\mathbf t}(\vec y) \text{ and } \vec x^{(r)}\ne \vec y^{(r)}] \le 2^{-r}.
\end{equation}
(Notice that the probabilities are with respect to the random bit strings $\mathbf s^{(1)},\mathbf s^{(2)}\dotsc$, which determine $\mathbf t(\vec z)$ and $\vec z^{(r)}$.) This follows by considering the above procedure as a stochastic process. We restrict our attention only to the trajectory where the parities are equal and the resulting vectors are different, $t_k(\vec x)=t_k(\vec y)$ and $\vec x^{(k)}\ne \vec y^{(k)}$. The transition probability along this path at step $k$ is bounded by
\begin{multline}
  \mathbb P[ t(\mathbf s^{(k)};\vec x^{(k-1)})= t(\mathbf s^{(k)};\vec y^{(k-1)}) \text{ and }
  f(\mathbf s^{(k)};\vec x^{(k-1)})\ne f(\mathbf s^{(k)};\vec y^{(k-1)})]
  \\\le
  \mathbb P[ t(\mathbf s^{(k)};\vec x^{(k-1)})= t(\mathbf s^{(k)};\vec y^{(k-1)})].
\end{multline}
Since we have $\vec x^{(k-1)}\ne \vec y^{(k-1)}$ from the previous step, the probability to obtain equal parity bits for $\vec x^{(k-1)}$ and $\vec y^{(k-1)}$ is exactly $\frac12$.

As announced before, the parity string $\mathbf t$ allows us now to determine the vector $\vec x^{(r)}$ with high probability if $r$ is chosen sufficiently large. This can only fail due to two reasons. Either the original vector $\vec x$ is not $\epsilon$-typical, or there is still more than one resulting vector that matches all the parities. Considering these two sources of failure, we obtain the total failure probability
\begin{equation}\begin{split}
  p_\text{fail}(n,\epsilon,r)
  &=     q_{n,\epsilon} + \mathbb P[ \vec X\in A_{n,\epsilon} \text{ and } \exists \vec y\in A_{n,\epsilon}\colon \mathbf t(\vec X)=\mathbf t(\vec y) \text{ and } \vec X^{(r)}\ne \vec y^{(r)}]
  \\&\le q_{n,\epsilon} + \sum_{\vec x,\vec y\in A_{n,\epsilon}} \mathbb P[\vec X=\vec x]\, \mathbb P[\mathbf t(\vec x)=\mathbf t(\vec y) \text{ and } \vec x^{(r)}\ne \vec y^{(r)}]
  \\&\le q_{n,\epsilon} + 2^{nH(X)+n\epsilon} \, 2^{-r}.
\end{split}\end{equation}
where we abbreviated $q_{n,\epsilon} = \mathbb{P}[\vec X \notin A_{n,\epsilon}]$ and we used the upper bound on the number of elements in $A_{n,\epsilon}$ mentioned above, see also Appendix~\ref{sec_eps_typical}. We choose now some rate $0<R<1-H(X)$ and some $0<\epsilon< 1-H(X)- R$. With $r^*_n=\floor{(1-R)n}$ and $\delta=1-H(X)-R-\epsilon>0$, we get $p_\text{fail}(n,\epsilon,r^*_n)< q_{n,\epsilon}+ 2^{1-n\delta}$. Due to $\lim_{n\to\infty} q_{n,\epsilon}= 0$, we can choose $n$ so large that $p_\text{fail}(n,\epsilon,r_n^*)$ falls below any fixed threshold $p^*>0$. At the same time, the effective rate of remaining qubit pairs $R^*=\frac1n(n-r^*_n)$ is lower bounded by $R^*\ge R$.

Lemma~\ref{lemm_hashing_exact} follows now from these considerations by choosing $R=\frac12(1-h)$ and $\epsilon=\frac14(1-h)$. Here, we abbreviated $h=S(\rho)$ and we use $r=\floor{\frac n2(1+h)}$. Then $n=n(h,F)$ has to be such that $1-p_\text{fail}(n,\epsilon,r)\ge F$. Consequently, the LOCC protocol in Lemma~\ref{lemm_hashing_exact} is given by the following.
\begin{enumerate}
\item
Collect $n=n(h,F)$ qubit pairs $\rho^{\otimes n}$.
\item
Generate $r=\floor{\frac n2(1+h)}$ random bit strings $\mathbf s_1,\dotsc,\mathbf s_r$ of length $n, n-2,\dotsc,n-2r+2$, respectively.
\item \label{step:proto}
Sequentially measure the parity bits $t_1,\dotsc, t_r$ corresponding to $\mathbf s_1,\dotsc,\mathbf s_r$ using LOCC. This destroys $r$ qubit pairs.
\item
Determine a sequence $\vec y\in A_{n,\epsilon}$ with $\epsilon=\frac14(1-h)$ such that $\ket{\psi(\vec y)}$ in the previous step would give the same parities. If no such sequence exists, choose an arbitrary sequence.
\item
Determine $\vec x$ from the sequence $\vec y$ by computing the final state $\ket{\psi(\vec x)}$ of Step~\ref{step:proto} under the assumption that the initial state was $\ket{\psi(\vec y)}$.
\item
Perform the local unitaries achieving $\ket{\psi(\vec x)}\mapsto \ket {\Phi^+}^{\otimes m}$.
\end{enumerate}
It remains to show that $\tr(\sigma_+^{\otimes m}\rho')\ge F$ for the final state $\rho'$ of the protocol. This is the case, since if $\vec x$ is guessed correctly, then the resulting state is $\sigma_+^{\otimes m}$ while otherwise it is some other, unknown $m$-qubit state $\tau$. Hence, $\rho'=(1-p_\text{fail})\sigma_+^{\otimes m} + p_\text{fail}\tau$ and the assertion follows.

\section{Typical sets}
\label{sec_eps_typical}

At the core of the classical part of the hashing protocol \cite{hashingoriginal, hashinglater}, see Section~\ref{sec_hashing}, is a classical result from information theory, see, for example, Theorem~3.1.2 in Ref.~[\onlinecite{cover_2005}]. According to this result, roughly $H(X)$ bits suffice to describe with high probability the value of a random source $X$. In detail, considering $n$ i.i.d.\ random variables $\vec{X} = ( X_1, X_2, \dotsc, X_n)$, a message $\vec{x}=(x_1,x_2,\dotsc,x_n)$ is $\epsilon$-typical, in symbols, $x \in A_{n, \epsilon}$, if its probability of occurrence $P(\vec{x}) = \mathbb{P}( \vec{X} = \vec{x}) $ fulfills
\begin{equation}\label{eq_typical_def}
  2^{-n ( H(X) + \epsilon)} \leq P( \vec{x} ) \leq 2^{-n ( H(X) - \epsilon )}.
\end{equation}
From this definition, we immediately obtain an upper bound on the size $\abs {A_{n,\epsilon}}$ of the $\epsilon$-typical set.
\begin{theorem}\normalfont
The number of $\epsilon$-typical messages is bounded by
\begin{equation}
  \abs{ A_{n,\epsilon}} \leq 2^{n(H(X) + \epsilon )}.
\end{equation}
\end{theorem}
\begin{proof}
Each message has at least probability $2^{-n(H(X)+\epsilon)}$. From the normalization of probabilities we obtain $\abs{ A_{n,\epsilon}} \,2^{-n(H(X)+\epsilon)}\le 1$ and the assertion follows.
\end{proof}
Consequently, $nH(X)-n\epsilon$ bits are sufficient to label each $\epsilon$-typical message. Now we quantify the probability $\mathbb{P}( \vec{X} \notin A_{n,\epsilon})$ for a message to be not $\epsilon$-typical. To this end, we rephrase the definition of $\epsilon$-typicality by using
\begin{equation}
  \log_2 P(\vec{x}) = \sum_{j=1}^n \log_2 P(x_j).
\end{equation}
Then Eq.~\eqref{eq_typical_def} reads
\begin{equation}\label{App_EpsTyp}
  \left\lvert - \frac{1}{n} \sum_{j=1}^n \log_2 P(x_j) - H(X) \right\rvert \le \epsilon.
\end{equation}
Moreover we consider the probability of an outcome as a random variable $P(X)$ and note that
\begin{equation}\label{App_Exp_Shannon}
  \mathbb{E} [ \log_2 P(X)] = \sum_x \mathbb{P}(X=x) \log_2 P(x) = -H(X).
\end{equation}

\begin{theorem}\normalfont
The probability that a message is not $\epsilon$-typical converges to zero,
\begin{equation}
\lim_{n\to\infty} \mathbb{P} ( \vec{X} \notin A_{n, \epsilon} )= 0.
\end{equation}
\end{theorem}
\begin{proof}
By using Eq.~\eqref{App_EpsTyp} and Eq.~\eqref{App_Exp_Shannon} we rephrase the statement about $\epsilon$-typical messages as a statement about the arithmetic mean and the expectation value of i.i.d.\ random variables,
\begin{equation}\begin{split}
  \mathbb{P} ( \vec{X} \notin A_{n,\epsilon} )
  &= \mathbb{P} \left( \left\lvert - \frac{1}{n} \sum_{j=1}^n \log_2 P(X_j) - H(X) \right\rvert > \epsilon \right) \\
  &= \mathbb{P} \left( \left\lvert \frac{1}{n} \sum_{j=1}^n \log_2 P(X_j) - \mathbb{E} [ \log_2 P(X) ] \right\rvert > \epsilon \right)
\end{split}\end{equation}
From the weak law of large numbers it follow that the last expression converges to $0$ as $n\to\infty$.
\end{proof}

\end{document}